\documentclass[preprint]{elsarticle}
\journal{CAD}
\renewcommand*{\today}{April, 2016}

\usepackage[utf8]{inputenc}
\usepackage[T1]{fontenc}
\usepackage[pdftex,hyperfootnotes=false,pdfpagelabels,hidelinks]{hyperref}
\usepackage{amsthm}
\usepackage{amsfonts}
\usepackage{amsmath}

\usepackage{caption}
\usepackage{subcaption}
\usepackage{cleveref}
\newtheorem{lemma}{Lemma}
\usepackage{nicefrac}

\usepackage{fancyhdr}
\newcommand{\draftfig}{.small}
\newcommand{\SK}{\mathcal{S}}
\newcommand{\WF}{\mathcal{W}}
\newcommand{\R}{\mathbb{R}}
\newcommand{\bigO}{\ensuremath{\mathcal{O}}}
\def\+{\discretionary{}{}{}}

\graphicspath{{./figs/}}

\begin{document}

\begin{frontmatter}
\title{Straight Skeletons with Additive and Multiplicative Weights and Their
       Application to the Algorithmic Generation of Roofs and Terrains}

\author[sbg]{Martin Held}
\author[sbg]{Peter Palfrader}

\address[sbg]{
    Universit\"at Salzburg,
    FB Computerwissenschaften,
    5020 Salzburg, Austria;
    \texttt{\{held,\+palfrader\}@cosy.sbg.ac.at}.
}

\begin{abstract}
  We introduce additively-weighted straight skeletons as a new generalization
  of straight skeletons.  An additively-weighted straight skeleton is the
  result of a wavefront-propagation process where, unlike in previous variants
  of straight skeletons, wavefront edges do not necessarily start to move at
  the begin of the propagation process but at later points in time. We analyze
  the properties of additively-weighted straight skeletons and show how to
  compute straight skeletons with both additive and multiplicative weights.

  We then show how to use additively-weighted and multiplicatively-weighted
  straight skeletons to generate roofs and terrains for polygonal shapes such
  as the footprints of buildings or river networks. As a result, we get an
  automated generation of roofs and terrains where the individual facets have
  different inclinations and may start at different heights.
\end{abstract}

\begin{keyword}
weighted straight skeleton \sep additive weights \sep multiplicative weights
\sep roof \sep terrain \sep wavefront propagation
\end{keyword}

\end{frontmatter}

\section{Introduction}

\subsection{Motivation and prior work} \label{sec:prior_work}

Straight skeletons were introduced to computational geometry over 20
years ago by Aichholzer et al.~\cite{Aic&95}. Suppose that the edges
of a simple polygon $P$ move inwards with unit speed in a
self-parallel manner, thus generating mitered offsets inside of $P$.
Then the (unweighted) straight skeleton of $P$ is the geometric graph
whose edges are given by the traces of the vertices of the shrinking
mitered offset curves of $P$, see \Cref{fig:sk} and \Cref{sec:prelim}.

Multiplicatively-weighted straight skeletons were first mentioned by
Aichholzer and Aurenhammer~\cite{AiAu98} and then by Eppstein and
Erickson~\cite{EpEr99}.  Roughly, the presence of multiplicative weights
implies that the edges of $P$ are allowed to move inwards at different speeds.
Recently, multiplicatively-weighted straight skeletons were studied in
detail by Biedl et al.~\cite{Bie&15a} who analyzed under which conditions
properties of the unweighted skeleton carry over to the weighted pendant.

Unweighted and multiplicatively-weighted straight skeletons are known to have
applications in diverse fields.
Aurenhammer~\cite{Aurenhammer08} investigates fixed-share decompositions of
convex polygons using skeletons with specific positive multiplicative weights.
Barequet et al.~\cite{Bar&08} employ multiplicatively-weighted straight
skeletons as a theoretical tool for computing (unweighted) straight skeletons
in three-space. Barequet and Yakersberg \cite{BaYa03} morph shapes by means of
their straight skeletons.
Tomoeda and Sugihara use straight skeletons to create signs with an illusion of
depth~\cite{ToSu12}, and
Sugihara also uses multiplicatively-weighted skeletons in the design of pop-up
cards~\cite{Sugihara13}.
Haunert and Sester~\cite{HaSe08} apply them for topology-preserving area
collapsing in geographic information systems (GIS). In another GIS
application, Vanegas et al.\ \cite{Van*12} use straight skeletons for
generating parcels in urban modeling.

The automatic generation of roofs of buildings based on straight skeletons of
their footprints (i.e., bird's eye view) has also received
wide-spread attention in large-scale urban modeling. E.g., Larive and Gaildrat
\cite{LaGa06}, M\"uller et al.\ \cite{Mue*06}, and Buron et al.\ \cite{Bur*13}
combine GIS data and shape grammars with production rules to generate roofs
for buildings. As a starting point or if a purely grammar-based generation is
not possible, they resort to roofs obtained from straight skeletons. The roofs
in the recent work by Sugihara \cite{Sugi13a,Sugi15} are based on straight
skeletons as well.  Furthermore, Laycock and Day~\cite{LaDa03} and Kelly and
Wonka~\cite{KeWo11} use multiplicatively-weighted straight skeletons for
modeling roofs in more realistic ways.

A problem closely related to the generation of roofs is the
(re-)construction of terrains. For instance, we might be given
a river map together with estimates of the slopes of the terrain. Straight
skeletons offer a promising approach to both roof generation and terrain
construction.

Roofs created by straight skeletons are limited to hip roofs and, with some
postprocessing, gable roofs. Their ridges tend to be parallel to long
edges of the footprint of the building. Typically, such roofs will not have
ridges that are perpendicular to long (parallel) edges of the footprint.

\subsection{Our contribution}

We introduce an additively-weighted straight skeleton as a new
generalization of straight skeletons: If additive weights are present then
edges of the input need not all start to move at the same time.  We analyze
the properties of additively-weighted straight skeletons and show how to
extend the standard algorithmic framework for computing straight skeletons
(based on wavefront propagation) to additively-weighted straight skeletons.

We also argue that this framework allows to handle both additive and
multiplicative weights.
Multiplicative weights translate to different speed functions for the input
edges, but each speed stays constant throughout the entire movement of the
edge. As a matter of fact, in our framework any speed function that remains
piecewise constant could be used for an edge, thus extending traditional
straight skeletons even further.

The input for our algorithm need not be constrained to simple
polygons. Rather, any so-called planar straight-line graph (PSLG) forms a
permissible input. (A PSLG is a collection of straight-line segments that do
not intersect pairwise except at common end-points.)

Combining both additive and multiplicative weights yields input edges
that (1) are allowed to move at different speeds and (2) may start at
different times.  As a result, we get an automated generation of roofs
or terrains where the individual facets have different inclinations
and may start at different heights.  In particular, additive weights
allow for gable roofs without postprocessing, with the ridge being
perpendicular to some long edge of the footprint. General piecewise
constant speed functions result in piecewise linear surfaces (roof,
terrain, etc.) where individual facets may have kinks.

As for unweighted straight skeletons, additively and multiplicatively
weighted straight skeletons come with an important property: A
raindrop that hits a facet of a surface generated by means of a
weighted straight skeleton is guaranteed to run off. That is, no local
minima can occur on the surface.

\section{Preliminaries} \label{sec:prelim}

\paragraph{Wavefront Propagation Process}

Let $P$ denote a simple polygon.  The straight skeleton of $P$ is
defined by means of a wavefront propagation process.  The wavefront
$\WF_P(t)$ is a set of wavefront polygons and changes with time $t$.
Initially, at time zero, $\WF_P(0)$ consists only of $P$.  Then, as time
increases, the edges of $\WF_P(t)$ move towards the interior of $P$ at unit
speed in a self-parallel manner, thereby preserving incidences.  Thus, the
vertices of $\WF_P(t)$ move along the angular bisectors of polygon edges, and
the wavefront corresponds to a mitered offset of $P$, see \Cref{fig:sk}.

\begin{figure}[htb!]
  \centering
  \includegraphics[page=1]{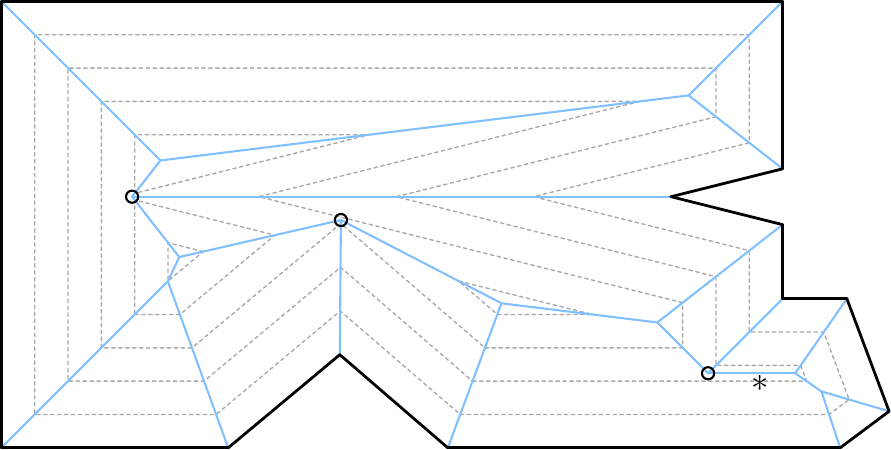}
    \caption{%
       Polygon (bold) with its straight skeleton.
       A family of mitered offset curves, i.e., the wavefront at different
       times, is shown in dotted gray.
       The straight skeleton nodes marked with $\circ$ are the result of
       split events; the others come from edge events.
       The straight skeleton arc marked with $*$ is one that was added
       when two parallel wavefront edges moved into each other.
     }
  \label{fig:sk}
\end{figure}

To maintain the planarity of the wavefront during the propagation process,
Aichholzer et al.~\cite{Aic&95} resolve non-planarities when they occur:
\begin{itemize}
  \item
    In an \emph{edge event}, an edge of the wavefront has shrunk to zero
    length.  This edge is removed from the wavefront, resulting in the two
    adjacent edges to become neighbors.
  \item In a \emph{split event}, a reflex vertex $v$ reaches another
    part of the wavefront.  (A vertex $v$ of $P$ is called reflex if
    the interior angle at $v$ is greater than $180^{\circ}$, and
    convex if it is less than $180^{\circ}$; tangential vertices with
    interior angle equal to $180^{\circ}$ can be ignored during the
    wavefront propagation.) The wavefront is split at this locus, and
    two separate polygons replace the previous polygon to restore
    planarity of the wavefront after the event.  Typically this will
    happen when $v$ reaches the interior of a wavefront edge.
    However, if $v$ reaches another vertex then more complex
    interactions are possible, resulting in non-elementary
    events~\cite{BieHP14b}.
\end{itemize}
Since the wavefront moves inwards within a polygon of finite extension, at
some point $\bar{t}$ in time all wavefront polygons will have collapsed, thus
resulting in $\WF_P(\bar{t})$ being the empty set. At this time $\bar{t}$ the
propagation process ends.

\paragraph{Straight Skeleton}

The \emph{straight skeleton} $\SK(P)$ is the geometric graph whose
edges are the traces of all vertices of $\WF_P(t)$ over the entire propagation
period.
In addition, if two parallel wavefront edges move into each other during the
wavefront propagation, then the portion common to them is added to the straight
skeleton as well while the portions that belong to only one of them remain in
the wavefront~\cite{Bie&15a}.
\Cref{fig:sk} shows wavefront polygons at different times and the resulting
straight skeleton.

To avoid ambiguities, one generally refers to the edges of the straight
skeleton as \emph{arcs} and reserves the term \emph{edges} for the input
polygon and the wavefront.  Likewise, the vertices of a straight skeleton
are called \emph{nodes}.

The straight skeleton of a polygon is a tree and
each interior node of $\SK(P)$ is of degree three for input in general
position such that only elementary edge and split events occur during the
wavefront propagation~\cite{Aic&95}.
Since the vertices of the wavefront move along angular bisectors of edges of
$P$, all arcs of the straight skeleton are straight-line segments.

\paragraph{Faces}

The wavefront fragments of the polygon edge $e$ at time $t$ are
contained in $\overline{e} + t \cdot n_e$, where $\overline{e}$ is the
supporting line of $e$ and $n_e$ is its inward facing unit normal.  We
denote by $e(t)$ the (possibly empty) set of these wavefront fragments
of edge $e$ at time $t$.  Every \emph{face} of the straight skeleton is
traced out by the fragments of exactly one input edge over time, i.e.,
$f(e) := \bigcup_{t \ge 0} e(t)$ for the face $f(e)$ of edge $e$. Furthermore,
it is known that $f(e)$ is monotone  with respect to $e$ \cite{Aic&95}.

\paragraph{Roof Model}

The roof model~\cite{Aic&95} raises the wavefront propagation into
three-space, with the third ($z$-)coordinate being the time $t$.  With $P$
embedded in the $xy$-plane $t=0$, the propagation of the wavefronts over time
forms a polytope over $P$.  This piecewise linear and continuous polytope
$R(P) := \bigcup_{t \ge 0} (\WF_P(t)\times \{t\})$ is called the \emph{roof}
of $P$.  This roof is a terrain, i.e., it is a $z$-monotone surface where each
line parallel to the $z$-axis intersects it at most once.

The roof model is a useful theoretical tool when dealing with straight
skeletons as it makes some proofs easier.  It is also directly useful
as a solution for modeling terrains or actual roofs of buildings.
See \Cref{fig:sk-roof} for an illustration.

\begin{figure}[ht!]
  \centering
  \includegraphics[page=1,width=0.7\textwidth]{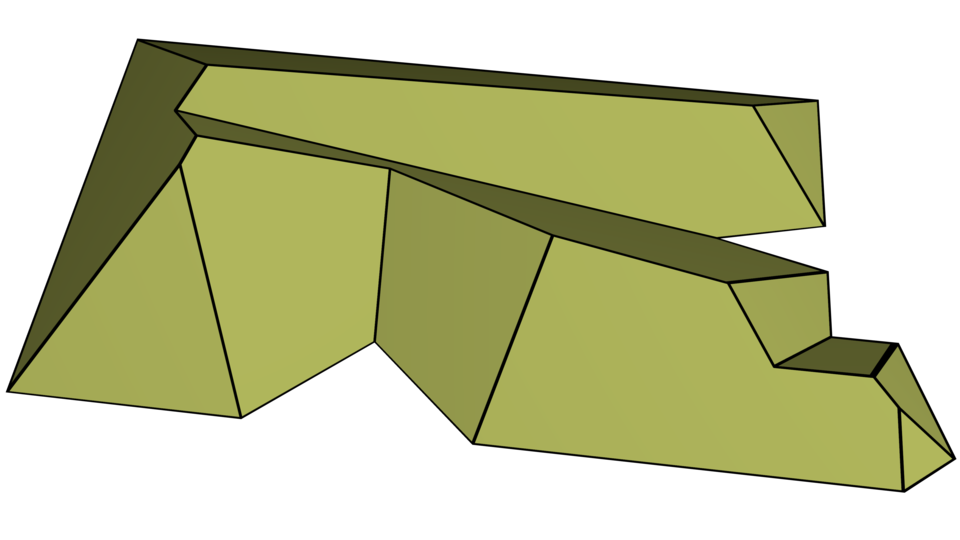}
    \caption{%
       The roof induced by the straight skeleton from \Cref{fig:sk}.
     }
  \label{fig:sk-roof}
\end{figure}

\paragraph{Straight Skeleton with Multiplicative Weights}

As a first generalization of unweighted straight skeletons, the
multiplicatively-weighted straight skeleton was introduced early on by
Aichholzer and Aurenhammer~\cite{AiAu98} and then by Eppstein and
Erickson~\cite{EpEr99}.
In the presence of multiplicative weights, wavefront edges no longer move at
unit speed. Rather, they move at different speeds depending on a weight
function $\sigma \colon E \to \R$ where $E$ is the edge-set of $P$.
That is, every edge $e$ is assigned its own constant speed $\sigma(e)$.
The wavefront fragments of $e$ are contained in
$\overline{e} + t \cdot \sigma(e) \cdot n_e$. See \Cref{fig:sk-mw}
for a multiplicatively-weighted straight skeleton; the corresponding roof is
shown in \Cref{fig:sk-mw-roof}.

\begin{figure}[htb!]
  \centering
  \includegraphics[page=1]{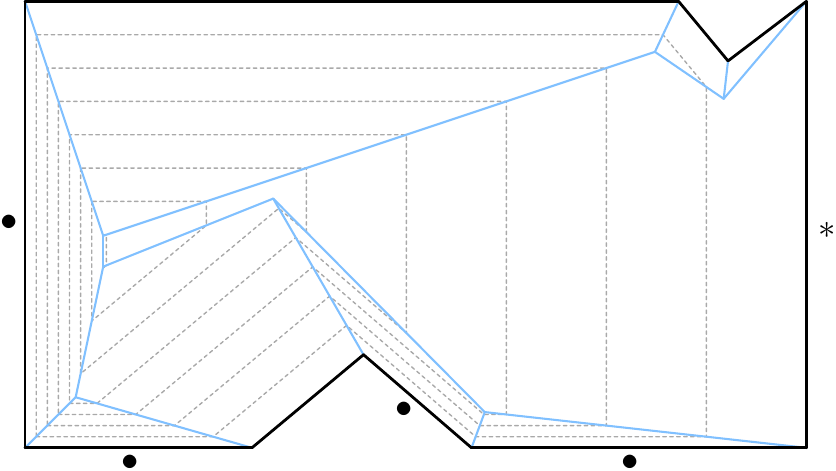}
    \caption{%
       Polygon with a multiplicatively-weighted straight skeleton
       and corresponding offsets.
       The edges marked with $\bullet$ and $*$ have a multiplicative
       weight of $\nicefrac{1}{3}$ and $3$, respectively.  All other
       edges have unit weight.
     }
  \label{fig:sk-mw}
\end{figure}

\begin{figure}[htb!]
  \centering
  \includegraphics[page=1,width=0.7\textwidth]{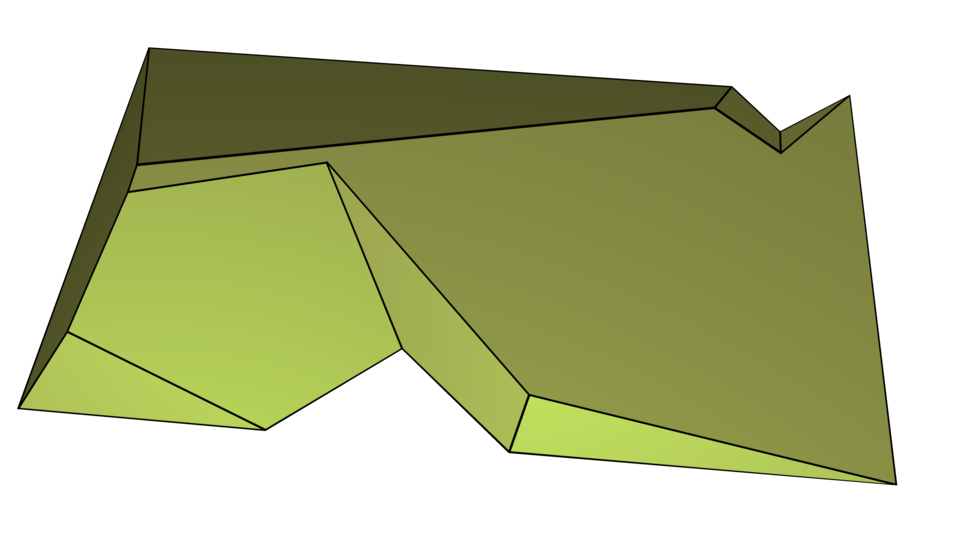}
    \caption{%
       The roof induced by the multiplicatively-weighted straight skeleton from
       \Cref{fig:sk-mw}.
     }
  \label{fig:sk-mw-roof}
\end{figure}

Although multiplicatively-weighted straight skeletons have been used
in various applications for quite a few years --- recall
\Cref{sec:prior_work} --- their characteristics were studied in full
mathematical rigor only recently by Biedl et al.~\cite{Bie&15a}.  If
all weights are required to be positive, then most of the well-known
properties of straight skeletons are preserved.  One prominent
exception is that a face need not be monotone to its defining input
edge any more; see, for example, the face traced out by the edge
marked with $*$ in \Cref{fig:sk-mw}.  For negative weights, $\SK(P)$
need not even be a tree, may contain crossings, and the roof need not
be a terrain~\cite{Bie&15a}. (A negative weight $\sigma(e)$ means that
the edge $e$ moves outwards with speed $|\sigma(e)|$.)  See
\Cref{fig:convexcrossing} for a multiplicatively-weighted straight
skeleton of a convex polygon that contains crossing arcs due to
negative weights.

\begin{figure}[htb!]
  \centering
  \includegraphics[page=1]{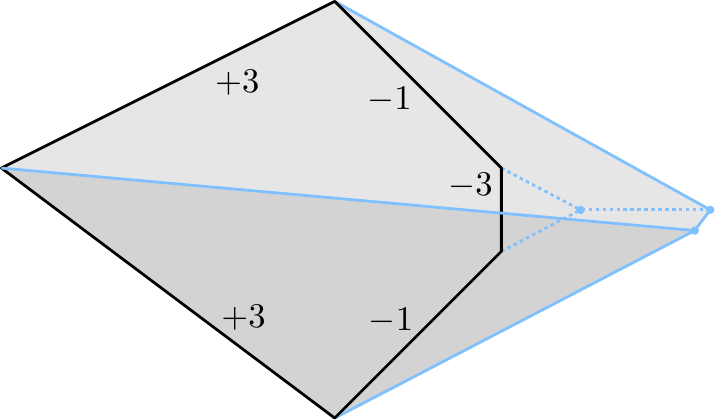}
    \caption{%
      We cannot allow negative multiplicative weights as then the straight
      skeleton may contain crossings~\cite{Bie&15a}.
     }
  \label{fig:convexcrossing}
\end{figure}

%%%%%%%%%%%%%%%%%%%%%%%%%%%%%%%%%%%%%%%%%%%%%%
%%%%%%%%%%%%%%%%%%%%%%%%%%%%%%%%%%%%%%%%%%%%%%
%%%%%%%%%%%%%%%%%%%%%%%%%%%%%%%%%%%%%%%%%%%%%%
\section{Additively-Weighted Straight Skeletons}

\subsection{Definition}

Given a simple polygon $P$ and an additive-weight function
$\delta\colon E \to \R^+_0$, we define the additively-weighted wavefront
$\WF_{P,\delta}(t)$ as follows: As in the unweighted case, $\WF_{P,\delta}(0)$
is identical to $P$.
However, wavefront edges do not all start to move immediately. Rather,
an edge of the wavefront that is emanated from polygon edge $e$
will only start to move inwards at unit speed at time $\delta(e)$.
Therefore, the wavefront fragments $e(t)$ of $e$ are contained in
\[
 \overline{e} + \begin{cases}
                 0                        & \mbox{if } t < \delta(e), \\
                 (t-\delta(e)) \cdot n_e  & \mbox{else.}
                \end{cases}
\]
This can also be written as $\overline{e} + \max(0, t-\delta(e)) \cdot n_e$.

Since wavefront edges no longer move all at once, wavefront vertices will not
travel exclusively along bisectors of input edges.  If both incident wavefront
edges have not yet started to move, then the wavefront vertex will obviously
remain stationary.  If exactly one incident wavefront edge has started to move,
then the wavefront vertex will travel on the supporting line of the other;
see \Cref{fig:v-on-edge}.

\begin{figure}[htb!]
  \centering
  \begin{subfigure}[b]{0.30\textwidth}
    \centering
    \includegraphics[page=1]{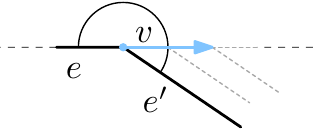}
    \caption[]{}
    \label{fig:v-on-edge:a}
  \end{subfigure}
  \hspace{0.1em}
  \begin{subfigure}[b]{0.30\textwidth}
    \centering
    \includegraphics[page=2]{v-on-edge}
    \caption[]{}
    \label{fig:v-on-edge:convex}
  \end{subfigure}
  \caption{%
    Vertex $v$ moving on the supporting line of wavefront edge $e$ that has not
    started to move yet.  The wavefront vertex $v$ is (a) reflex, (b) convex.
  }
  \label{fig:v-on-edge}
\end{figure}

During the propagation process the wavefront will see instances
of edge and split events, and it needs to be updated accordingly to restore
planarity after such an event.
Note that even edges and vertices that have not yet started to move can be
involved in both types of events. See, for instance, the edge in the top right
of \Cref{fig:skel-add}, which collapses before it starts moving. The wavefront
propagation process ends when all wavefront polygons have collapsed.

\begin{figure}[htb!]
  \centering
  \includegraphics[page=1]{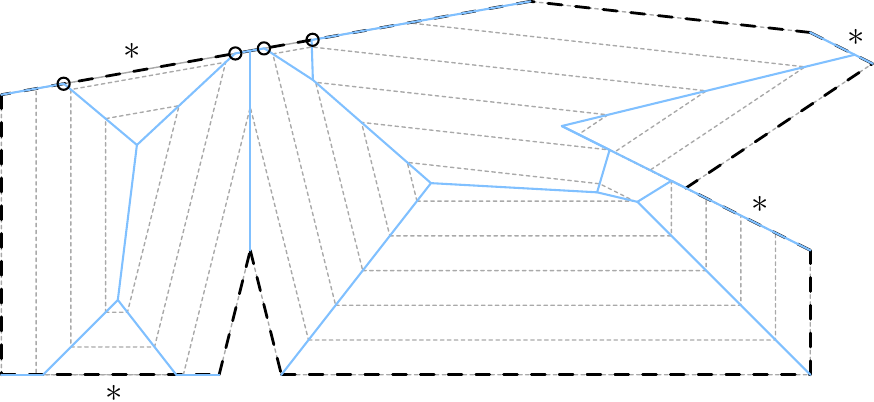}
    \caption{%
       Polygon (dashed) with an additively-weighted straight skeleton.
       The edges marked by $*$ have non-zero additive weights.
       A family of offset curves is shown in gray and dotted.
       The nodes marked with $\circ$ result from the speed-change event
       of the top left edge.
     }
  \label{fig:skel-add}

  \includegraphics[page=1,width=0.9\textwidth]{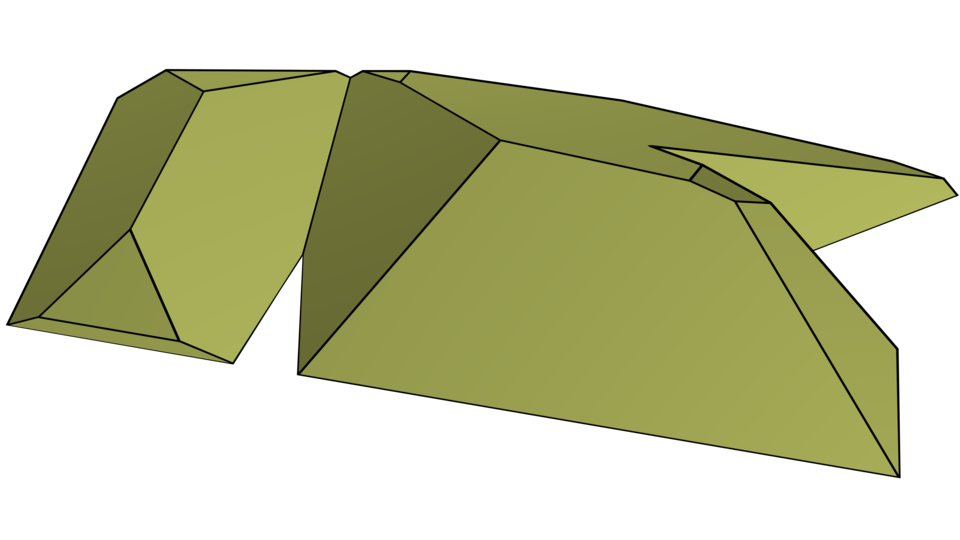}
    \caption{%
       Roof induced by the additively-weighted straight skeleton from
       \Cref{fig:skel-add}.
     }
  \label{fig:skel-add-roof}
\end{figure}

For simulating the propagation process, we consider an additional event:
We call the instance when an edge starts to move a \emph{speed-change event}.
For a wavefront edge emanated from polygon edge $e$ this will be at time
$\delta(e)$.  Speed-change events at time zero are called trivial, and we
usually may disregard them simply by setting up the kinetic wavefront with
already moving edges for time zero.

A speed-change event does not modify the combinatorial properties of the
wavefront, but it does change which elements move at which speeds.  In
particular, when a fragment starts to move, the direction and speed of its
incident wavefront vertices change.

The \emph{additively-weighted straight skeleton} $S(P, \delta)$ is then
defined as the geometric graph whose edges are the traces of vertices of
$\WF_{P,\delta}(t)$ over its propagation period.  Additionally, if two
parallel wavefront edges move into each other during the wavefront propagation,
then the portion common to them is also added to the straight skeleton.

As in the unweighted case, we call edges of $S(P, \delta)$ \emph{arcs}
and its vertices \emph{nodes}.
Similarly, we again call the loci traced out by the wavefront
segments $e(t)$ of edge $e$ the \emph{face $f(e)$ of $e$}, again given by
$f(e) := \bigcup_{t \ge 0} e(t)$.

\subsection{Properties}

\paragraph{Node Degrees}

In unweighted or multiplicatively-weighted straight skeletons, a node
will be of degree one when it is a leaf of the straight skeleton --- its
locus will then be at a vertex of the input polygon --- or of degree three
when it is the result of an elementary edge or split event.  Nodes with higher
degrees are also possible and are induced by non-elementary events where
more than three wavefront edges are involved~\cite{BieHP14b}.

In addition to these types of nodes, the additively-weighted straight skeleton
can have nodes of degree two. Degree-two nodes occur when a vertex of the
wavefront changes its velocity due to an incident wavefront edge starting to
move.  Such nodes always lie on the supporting line of the input edge whose
speed-change event caused it. See, for instance, the two pairs of nodes marked
with $\circ$ for the top left edge in \Cref{fig:skel-add} --- we have two
pairs rather than just one because the edge was involved in a split event
earlier on.

\paragraph{Nested Wavefronts}

The following \Cref{lem:nesteded_wavefronts} establishes that the wavefronts
of additively-weighted straight skeletons are nested inside each other. Hence
we get a situation similar to unweighted straight skeletons or straight
skeletons with positive multiplicative weights, except that the inclusions need not be proper,
since stationary edges of the wavefronts may overlap.

\begin{lemma}
  \label{lem:nesteded_wavefronts}
  Let $t_1,t_2\in\R_0^+$ with $t_1 < t_2$. Then $\WF_{P,\delta}(t_2)$ lies
  within $\WF_{P,\delta}(t_1)$, i.e.,
  $\WF_{P,\delta}(t_2) \subseteq \WF_{P,\delta}(t_1)$.
\end{lemma}
\begin{proof}
  Edges that are already moving at time $t_1$ keep moving towards the interior
  of $\WF_{P,\delta}(t_1)$, and edges that are still stationary do not move
  towards the outside either.

  Note that a vertex $v$ of a stationary edge $e$ may move.  However, if it
  does move, then it moves only towards the interior of the wavefront polygon
  $\WF_{P,\delta}(t_1)$ or along its boundary within the interior of $e$.  It
  never moves towards the outside of $\WF_{P,\delta}(t_1)$ since the other
  edge incident at $v$ does not move towards the outside of the wavefront
  polygon either.  This is illustrated in \Cref{fig:v-on-edge}.
\end{proof}

\paragraph{Roof Model}

The roof induced by an additively-weighted straight skeleton is defined similarly
to its unweighted pendant as
$R(P, \delta) := \bigcup_{t \ge 0} (\WF_{P,\delta}(t) \times \{t\})$.
Unlike the roof induced by an unweighted straight skeleton, it is clearly no
longer strictly $z$-monotone, since wavefront edges may stay on the same
supporting line during the propagation, resulting in vertical facets.  The
house depicted in \Cref{fig:house} has many such facets, namely the walls, as
all input edges have (different) additive weights assigned to them.  The
weight assigned to some edges is larger, resulting in some walls continuing
upwards while inclined roof facets already exist at the same height.

\begin{figure}[ht!]
  \centering

  \begin{subfigure}[b]{0.90\textwidth}
    \centering
    \includegraphics[page=1]{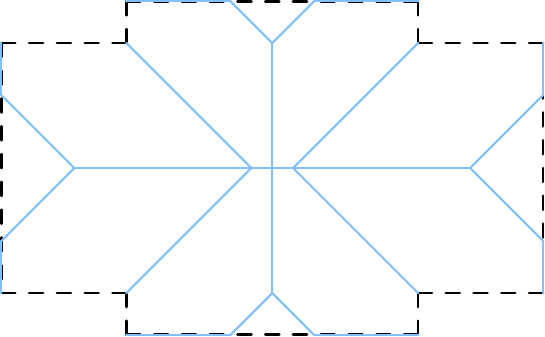}
    \caption[]{}
    \label{fig:house-floor-plan}
  \end{subfigure}

  \begin{subfigure}[b]{0.48\textwidth}
    \centering
    \includegraphics[page=3]{house-constructed-roof}
    \caption[]{}
    \label{fig:house-constructed-roof}
  \end{subfigure}
  \hspace{0.1em}
  \begin{subfigure}[b]{0.48\textwidth}
    \centering
    \includegraphics[page=1,width=1.0\textwidth]{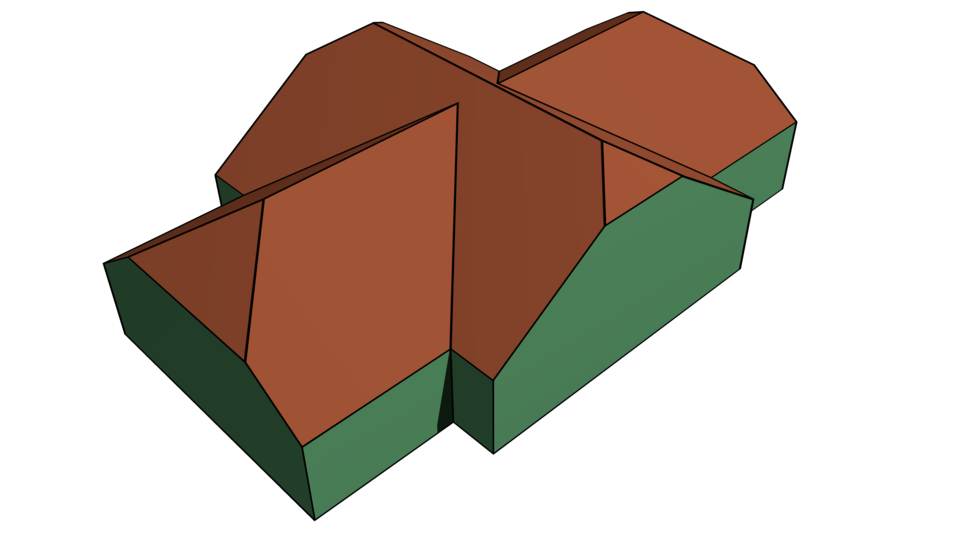}
    \caption[]{}
    \label{fig:house-render}
  \end{subfigure}

  \caption{%
    A house with a  roof induced by an additively weighted straight skeleton.
  }
  \label{fig:house}
\end{figure}

\begin{lemma}
  The roof $R(P, \delta)$ induced by an additively-weighted straight
  skeleton of $P$ is weakly $z$-monotone.
\end{lemma}
\begin{proof}
  This is a direct consequence of \Cref{lem:nesteded_wavefronts}.
\end{proof}

For each edge $e$ of the polygon, the roof will have at least one facet,
namely the one incident to $e$.  If the segments of $e$ see a speed-change
event during the propagation process, one additional facet per segment
will be visible in the roof.
Note that the total number of facets is still linear in the number of
vertices of the input polygon since additional wavefront segments can
only be caused by split events, whose number is guaranteed to be at
most linear in the number of vertices.

\paragraph{Crossings, Planarity, and Connectedness}

Biedl et al.~\cite{Bie&15a} show that the mul\-ti\-plica\-tive\-ly-weighted
straight skeleton is free of crossings%
\footnote{%
Briefly, the set of crossing-free embeddings of a graph $G$ is the closure of
the set of planar embeddings of $G$ with respect to the topology induced by a
vertex-wise distance metric.
} %
for positive edge weights $\sigma$:  Any locus $p$ of a crossing would have to be covered twice
by the wavefront.  Thus, the line parallel to $z$ through $p$ would intersect the roof twice.
However, the roof is strictly $z$-monotone, and therefore no such point can exist.

This result extends to additive weights including zero, as the projection of
the weakly $z$-monotone roof to the $xy$-plane is likewise free of crossings.

\begin{lemma}
   \label{lem:no-crossings}
   The additively-weighted straight skeleton of a simple polygon is free of
   crossings.
\end{lemma}

Note, however, that we cannot infer strict planarity from being free of
crossings:  Assume $v$ is a wavefront vertex where one incident edge $e$ has
not yet started to move.  Let $e'$ be the other edge incident at $v$.
Then $v$ travels on the supporting line of $e$, and the direction of
this movement depends on the angle that $e$ spans with $e'$; see
\Cref{fig:v-on-edge}.

Now let $e'$ collapse in an edge event.  Let $e''$ be the new neighbor
of $e$ and let $v'$ be the wavefront vertex common to $e$ and $e''$.
At the time of the edge event, when $v'$ first exists, it will be in the same
locus as $v$, the common vertex of $e$ and $e'$ that it replaces in the
wavefront.
If the angle at $v$ previously was reflex and at $v'$ it is now convex,
then $v'$ will move in the opposite direction of  $v$.
This results in the arc being traced out by $v'$ to overlap the arc
already traced out by $v$; see \Cref{fig:v-on-edge-backtrack}.

\begin{figure}[ht!]
  \centering
  \begin{subfigure}[b]{0.30\textwidth}
    \centering
    \includegraphics[page=3]{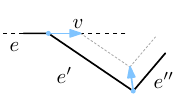}
  \end{subfigure}
  \hspace{0.1em}
  \begin{subfigure}[b]{0.30\textwidth}
    \centering
    \includegraphics[page=4]{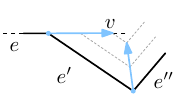}
  \end{subfigure}
  \hspace{0.1em}
  \begin{subfigure}[b]{0.30\textwidth}
    \centering
    \includegraphics[page=6]{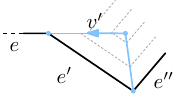}
  \end{subfigure}
  \caption{%
    After an edge event, a wavefront vertex may backtrack
    along an arc previously traced out.
  }
  \label{fig:v-on-edge-backtrack}
\end{figure}

Note that if the angle at $v$ was convex then $v$ moved on the supporting line
of $e$ within the interior of $e$ (\Cref{fig:v-on-edge:convex}).  The area
behind $v$ is outside of the wavefront polygon.  There are no events the
wavefront can undergo which would replace $v$ with a reflex vertex $v'$ that
moves in the opposite direction: such an event would cause $v'$ to move
outside of the wavefront polygon, violating \Cref{lem:nesteded_wavefronts}.
\Cref{fig:v-on-edge-no-backtrack} illustrates this fact.

\begin{figure}[ht!]
  \centering
  \includegraphics[page=8]{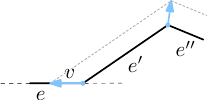}
  \caption{%
    For a convex vertex $v$ to be replaced by a reflex vertex $v'$ in
    an event, edge $e'$ would need to have a reflex vertex $v''$ with
    $e''$, with a sufficiently large interior angle at $v''$.
    However, such an edge $e'$ would then not shrink (since the arcs traced
    out by $v$ and $e' \cap e''$ diverge) and, therefore, the event
    would never happen.  }
  \label{fig:v-on-edge-no-backtrack}
\end{figure}

\begin{lemma}
 \label{lem:pathconnected}
  Let $\SK_{P, \delta}(t)$ be the portion of  $S(P, \delta)$ traced by
  the wavefront until time $t$, for some $t\ge 0$.
  If two points $p, q \in \SK_{P, \delta}(t)$
  are path-connected on $\SK_{P, \delta}(t) \cup \WF_{P,\delta}(t)$, then they
  are path-connected on $\SK(P, \delta)$.
\end{lemma}
\begin{proof}
  This is shown for multiplicatively-weighted straight skeletons in
  Lemma 13 of \cite{Bie&15a}.  We extend their proof here.

  Connectivity of $p$ and $q$ is not broken despite the changes to the
  wavefront caused by edge and split events.  Whenever an event removes
  elements from the wavefront, a straight-skeleton node is created to
  which arcs are incident that are traced out by vertices of each
  element involved.  Likewise, after a split event, all resulting
  wavefront components are connected to the straight skeleton node that
  witnessed that event.

  Therefore, it only remains to show that connectivity is maintained
  across speed-change events.
  This also holds, as such an event does not change the combinatorial
  properties of the wavefront.  It will only result in vertices
  moving at different velocities.  Thus, a path between points $p$ and $q$
  of $\SK_{P, \delta}(t)$ that goes over wavefront elements cannot be
  split by changes to the wavefront.
\end{proof}

\begin{lemma}
   \label{lem:connected}
   The additively-weighted straight skeleton of a simple polygon $P$ is
   connected.
\end{lemma}

\begin{proof}
  Initially, at time zero, the straight skeleton contains only
  disconnected nodes.  These are at the vertices of $P$ and
  therefore also at the wavefront vertices at time zero.
  The wavefront $\WF_{P,\delta}(0)$ is a single polygon identical
  to $P$, and so all nodes of $\SK_{P, \delta}(0)$ are connected via
  $\WF_{P,\delta}(0)$.
  By \Cref{lem:pathconnected}, therefore, $\SK(P, \delta)$ is connected.
\end{proof}

\paragraph{Faces}

As for unweighted and multiplicatively-weighted straight skeletons, for each
edge $e$ of $P$ we defined its face as $f(e) := \bigcup_{t \ge 0} e(t)$, where
$e(t)$ is the set of segments of the wavefront at time $t$ that were emanated
by $e$.  Initially, $e(0)$ will consist of only one segment that coincides
with $e$, but as the wavefront propagates, segments may get split and segments
may get dropped from $e(t)$ when they collapse.  However, at no time will a
segment just jump into existence.  Hence each face is connected.

Note, however, that $f(e)$ is not necessarily a simple polygon for edges
that do not immediately start to move.  The faces in \Cref{fig:skel-add} that
correspond to the edges with non-zero additive weights  demonstrate
this fact.  In clockwise order from the top left, we have a face whose
interior is disconnected, a face with an empty interior because its
corresponding edge collapsed before it started to move, and a face whose
interior is not adjacent to $e$ itself.

\begin{lemma}
  A face of an additively-weighted straight skeleton is connected.  Its
  interior may be empty or disconnected.  The interior need not be adjacent to
  the edge which emanated the face.
\end{lemma}

In the unweighted straight skeleton, the face of an edge $e$ is a
monotone polygon with respect to the supporting line of $e$.
This is, however, not always true for additively-weighted straight
skeletons:

\begin{lemma}
  A face of an additively-weighted straight skeleton need not be
  monotone with respect to the edge that emanated it.
\end{lemma}
\begin{proof}
  See the topmost face in \Cref{fig:skel-add}.
\end{proof}

\section{Discussion}

\subsection{Roof and Terrain Construction}

Once the footprint of the building upon which we want to construct a roof is
fixed, we have two major options to influence the shape of a
straight-skeleton-induced roof.

For each wall of the building we can specify at which height the roof
should start to slope inwards, if at all.  The additive weight $\delta(e)$ that
we assign to each input edge $e$ translates exactly to that height; recall
\Cref{fig:house}.  If $\delta(e)$ is sufficiently large such that during
the wavefront propagation its corresponding wavefront edge collapses before
that time, then the roof will not contain a corresponding sloped roof facet;
see \Cref{fig:simple-roof}. Note that the ridge of the roof of
\Cref{fig:simple-roof} is perpendicular to the two longer walls of the
house, thereby extending from one of them to the other. (This is a feature
which cannot be achieved by unweighted straight skeletons.)

\begin{figure}[htb!]
  \centering

  \begin{subfigure}[b]{0.48\textwidth}
    \centering
    \includegraphics[page=1]{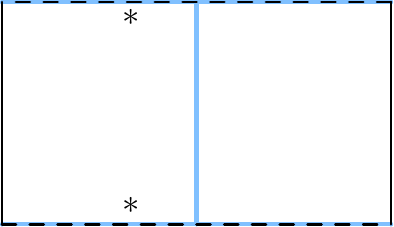}
    \caption[]{}
  \end{subfigure}
  \hspace{0.1em}
  \begin{subfigure}[b]{0.48\textwidth}
    \centering
    \includegraphics[page=1,width=1.0\textwidth,height=0.7\textheight,keepaspectratio]{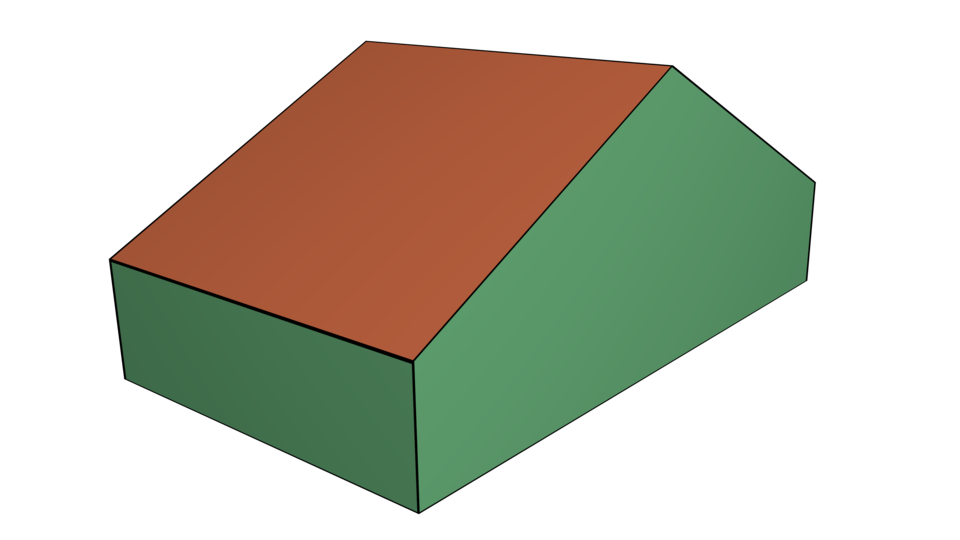}
    \caption[]{}
  \end{subfigure}

  \caption{%
    A very simple house: The footprint is just a rectangle.  The
    additive weights assigned to the two longer edges (marked with
    $*$) are sufficiently large such that the induced roof never has
    sloped facets corresponding to these input edges.  }
  \label{fig:simple-roof}
\end{figure}

By choosing appropriate multiplicative weights for the input edges, we
can influence the slopes of the facets of the roof. It is easy to see
that the tangent of the angle between a roof facet and the $xy$-plane
is inversely proportional to the multiplicative weight of the
corresponding input edge; see \Cref{fig:simple-m-roof}. Note that the
vertical walls of that house are not induced by the
multiplicatively-weighted straight skeleton. They could, however, be
generated by combining both additive and multiplicative weights! (See below.)

\begin{figure}[htb!]
  \centering

  \begin{subfigure}[b]{0.48\textwidth}
    \centering
    \includegraphics[page=1]{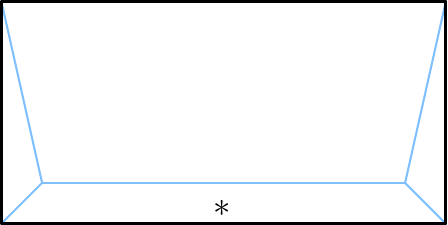}
    \caption[]{}
  \end{subfigure}
  \hspace{0.1em}
  \begin{subfigure}[b]{0.48\textwidth}
    \centering
    \includegraphics[page=1,width=1.0\textwidth]{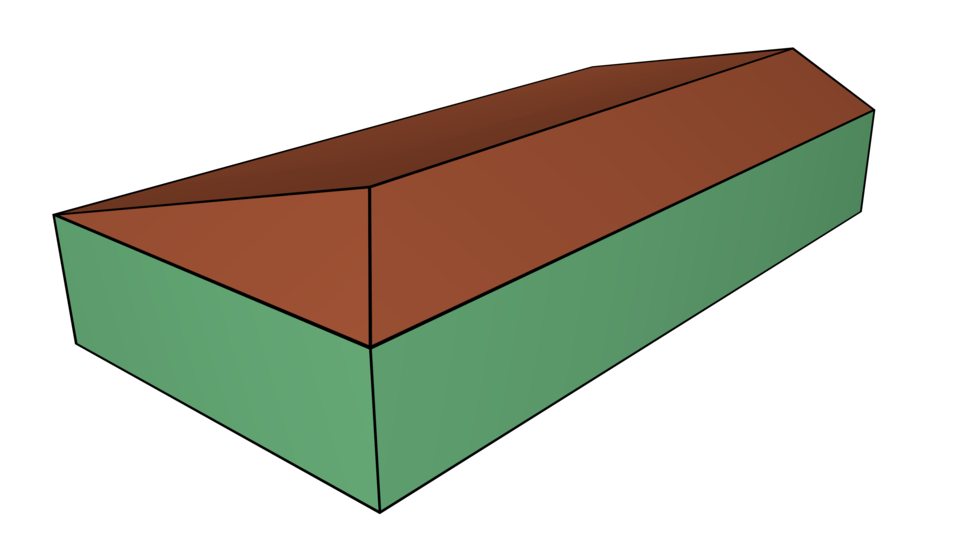}
    \caption[]{}
  \end{subfigure}

  \caption{%
    Again a simple footprint of a house.  The input edge marked with $*$ has
    a multiplicative weight less than that of the other edges.  Therefore,
    the corresponding facet in the roof is steeper.}
  \label{fig:simple-m-roof}
\end{figure}

\subsection{Generalizations}

Several generalizations seem natural.
First, we can combine multiplicative weights with additive weights.
Recall that in none of the arguments and proofs for the additively-weighted
straight skeleton we had relied on the speed of all moving edges to be equal
to one, or even just to be uniform.
As long as all speeds are non-negative, the multiplicatively- and
additively-weighted straight skeleton is well-defined and possesses the
properties laid out in the previous section.  Negative multiplicative weights
cannot be allowed since, e.g., \Cref{lem:nesteded_wavefronts} would not hold
and the straight skeleton need not even be a tree~\cite{Bie&15a}.

Negative additive weights, however, can be handled without much difficulty:
The wavefront propagation simply starts at the time that corresponds
to the smallest (negative) weight.
The straight skeleton itself is invariant to the addition of a constant value
to the additive weights of all edges --- it just results in the wavefront
propagation starting sooner or later. Hence, we could also handle negative
additive weights by subtracting the smallest negative weight from all weights,
thus running the wavefront propagation for only non-negative weights.  In
terms of the roof model, such additions merely mean shifting the whole
structure along the $z$-axis.

Second, as for unweighted straight skeletons, multiplicatively- and
additively-weighted straight skeletons can be defined not only for simple
polygons but for polygonal areas with holes or even for arbitrary planar
straight-line graphs (PSLGs). \Cref{fig:courtyard} shows a roof of a house
that has an inner courtyard. Supporting PSLGs is particularly important when
terrains defined by straight-line networks (rivers, roads, etc.) are to be
handled. Of course, since no interior or similar sideness is implied by a
PSLG, weights have to be specified for both sides of an edge of a PSLG.
See \Cref{fig:terrain} for a sample terrain induced by an additively-
and multiplicatively-weighted straight skeleton shown in \Cref{fig:terrain-sk}.

\begin{figure}[htb!]
  \centering
  \includegraphics[page=1,width=0.7\textwidth]{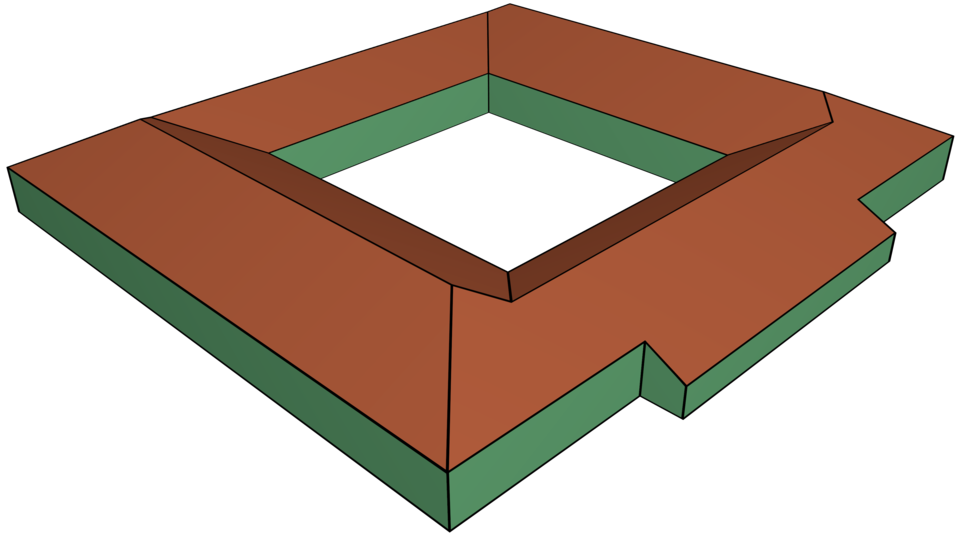}
    \caption{%
       The roof induced by an additively-weighted straight skeleton of a
       polygon with one hole.
     }
  \label{fig:courtyard}
\end{figure}

\begin{figure}[p]
  \centering

  \includegraphics[page=1,width=1.0\textwidth]{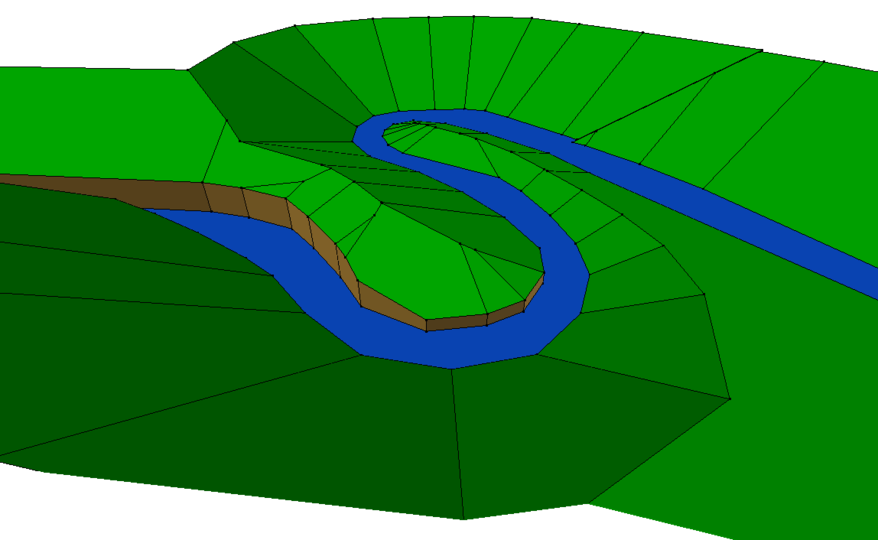}

  \caption{A terrain induced by a straight skeleton with both additive
    and multiplicative weights.  The corresponding weighted straight
    skeleton is shown in \Cref{fig:terrain-sk}.  }
  \label{fig:terrain}

  \vspace{3ex}

  \includegraphics[page=1]{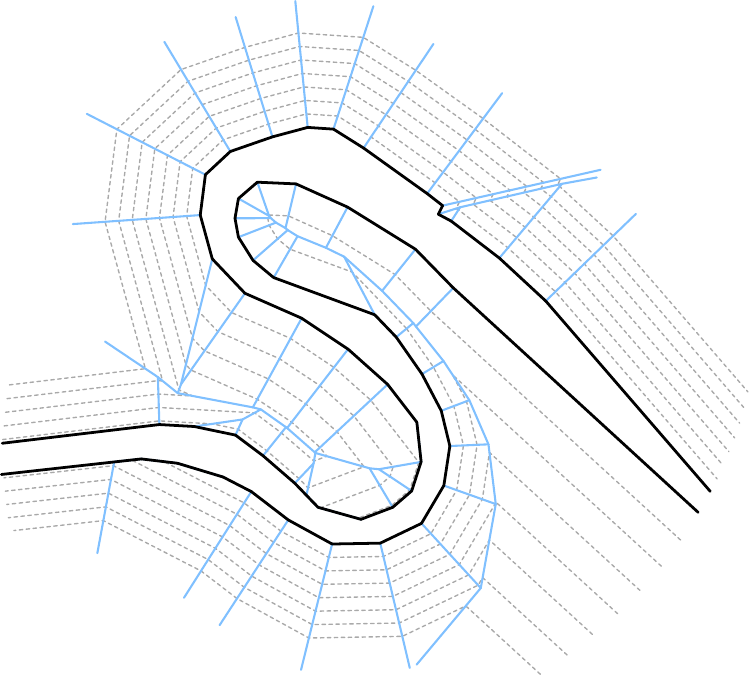}

  \caption{The planar straight-line graph (bold) whose straight skeleton (blue)
    induces the terrain from \Cref{fig:terrain}.
  }
  \label{fig:terrain-sk}
\end{figure}

Third, one can allow more than a single speed-change per edge.  As long as the
speed function for an edge remains piecewise constant, vertices will still
move along straight lines, all our lemmas and observations hold and, thus, the
straight skeleton will exhibit the properties discussed. See
  \Cref{fig:gablet} for a sample gablet roof where the additive and
  multiplicative weights of the edges change over time. Clerestory and
gambrel roofs can be generated similarly; see \Cref{fig:clerestory,fig:gambrel}.

Fourth, by inserting into the input polygon infinitesimally short edges near
edges with positive (or even infinite) additive weights, roof facets which
are not orthogonal to any input edge of non-zero length can be created.  This
variant was used to create the rhombic roof from \Cref{fig:rhombic}.

\begin{figure}[htb!]
  \centering
  \begin{subfigure}[b]{0.48\textwidth}
    \centering
    \includegraphics[page=1,width=0.95\textwidth]{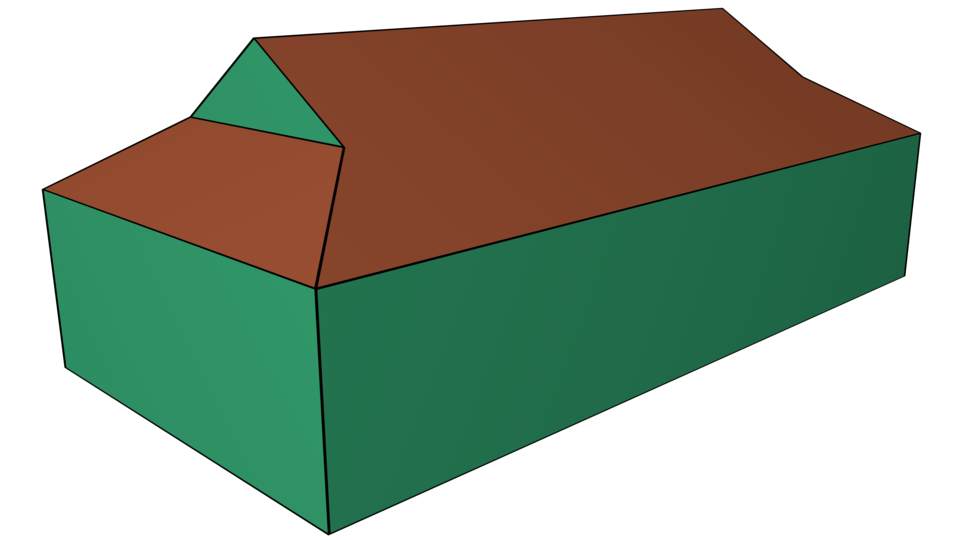}
    \caption[]{A gablet roof.}
    \label{fig:gablet}
  \end{subfigure}
  \hspace{0.1em}
  \begin{subfigure}[b]{0.48\textwidth}
    \centering
    \includegraphics[page=1,width=0.95\textwidth]{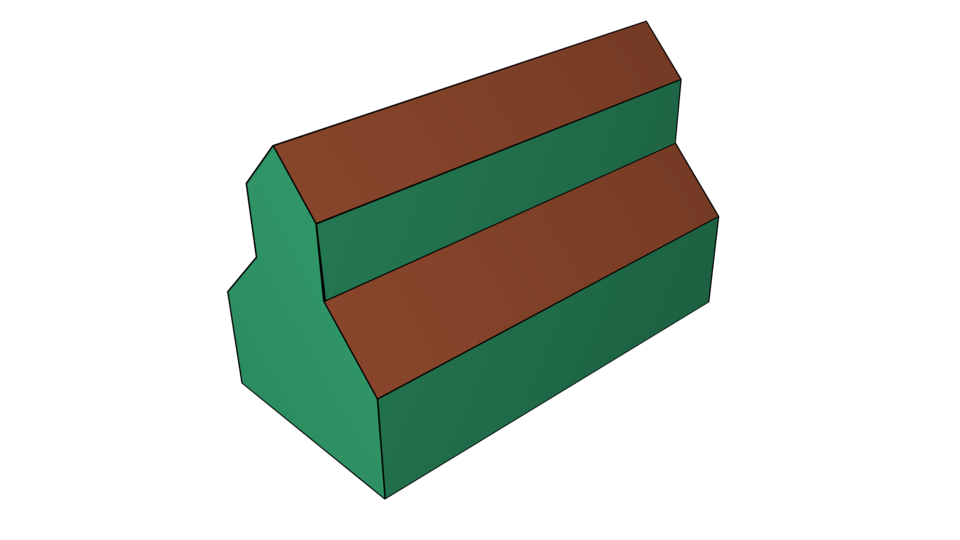}
    \caption[]{A clerestory roof.}
    \label{fig:clerestory}
  \end{subfigure}
  \\
  \begin{subfigure}[b]{0.48\textwidth}
    \centering
    \includegraphics[page=1,width=0.95\textwidth]{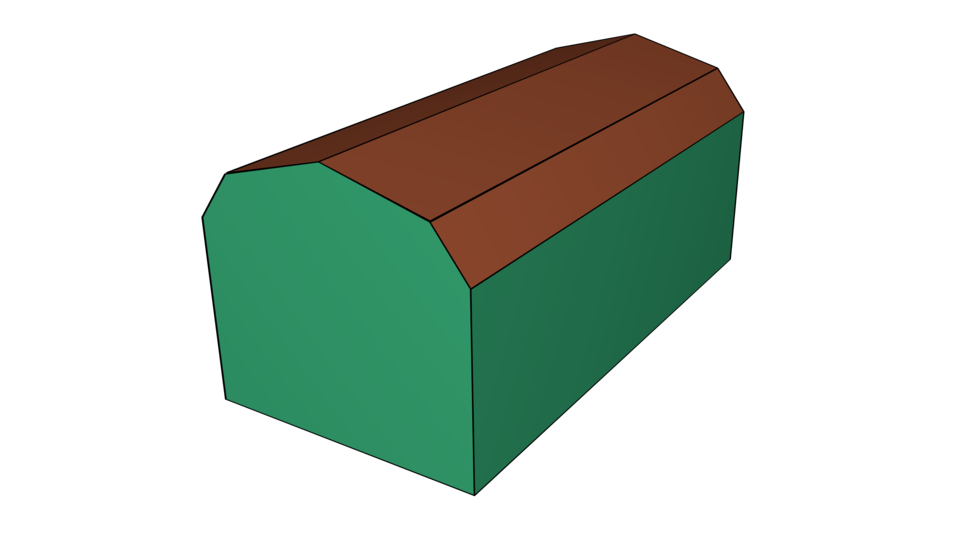}
    \caption[]{A gambrel roof.}
    \label{fig:gambrel}
  \end{subfigure}
  \hspace{0.1em}
  \begin{subfigure}[b]{0.48\textwidth}
    \centering
    \includegraphics[page=1,width=0.95\textwidth]{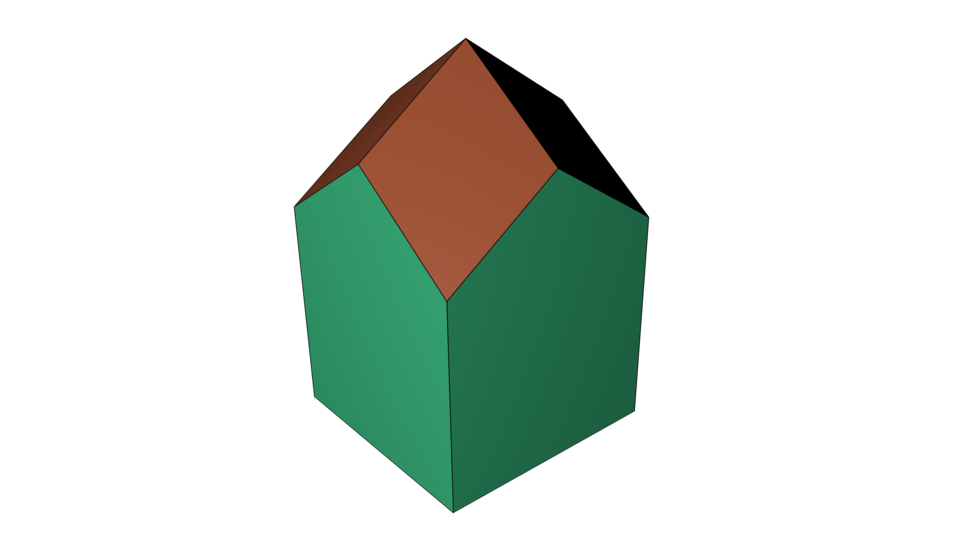}
    \caption[]{A rhombic roof.}
    \label{fig:rhombic}
  \end{subfigure}
  \\
  \begin{subfigure}[b]{0.48\textwidth}
    \centering
    \includegraphics[page=1,width=0.95\textwidth]{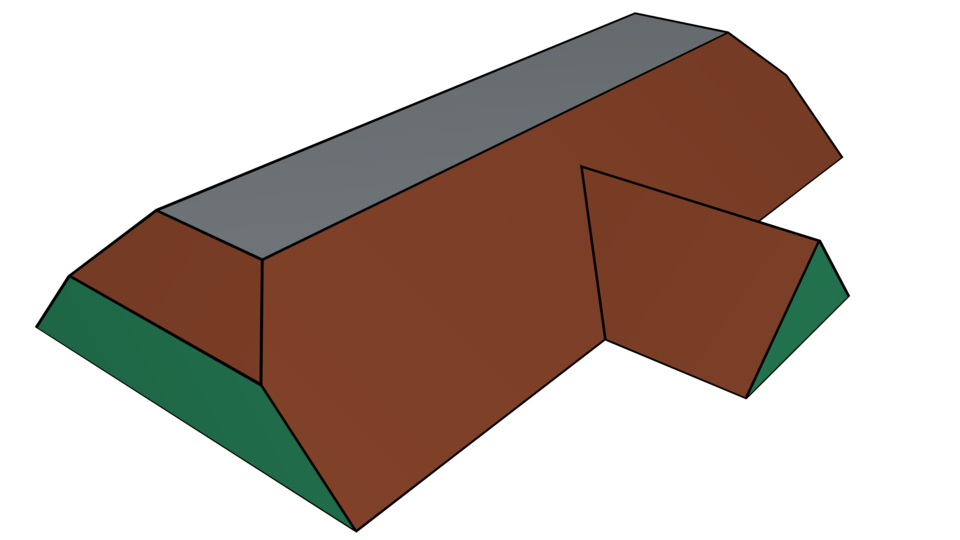}
    \caption[]{A flat-top roof.}
    \label{fig:flat_roof}
  \end{subfigure}
  \hspace{0.1em}
  \begin{subfigure}[b]{0.48\textwidth}
    \centering
    \includegraphics[page=1,width=0.95\textwidth]{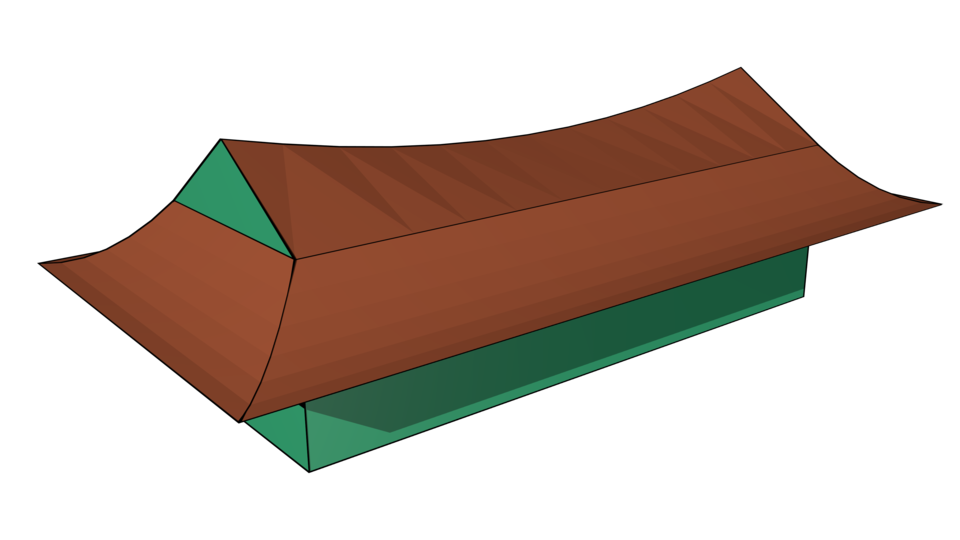}
    \caption[]{An Asian-style temple roof.}
    \label{fig:japanese_temple}
  \end{subfigure}

  \caption{%
    Different roof shapes can be created without additional
    post-processing (except for the temple roof in (f)), using an
    algorithmic framework based on weighted straight skeletons.  }
  \label{fig:roof-shapes}
\end{figure}

Recall that straight skeletons are defined by a wavefront
propagation. Running that process until all wavefronts have collapsed will
result in the full roof. We could, however, also halt the wavefront
propagation at some specific point in time. This will result in roofs which
have flat portions. See \Cref{fig:flat_roof} for a T-shaped house with a flat
top. From a mathematical point of view, a flat portion of a roof corresponds
to a speed change where the speed of one or more edges of the wavefront
becomes infinite.

Finally, standard post-processing techniques are also applicable to roofs
induced by straight skeletons. \Cref{fig:japanese_temple} shows a modification
of the gablet roof of \Cref{fig:gablet} where some ridges were replaced by
curved arcs, thus making it resemble the roof of an Asian-style temple.

\subsection{Computational Aspects}

A simple method is described by Aichholzer et al.~\cite{Aic&95} to compute the
unweighted straight skeleton: Compute the $\bigO(n)$ many collapse times of
all edges and the $\bigO(n^2)$ many times of all potential split events,
and maintain them in a priority queue.  When handling an event only a
constant number of edge collapses have to be recomputed, at constant cost each.
Since the total number of events is linear, the overall algorithm runs in
$\bigO(n^2 \log n)$ time within $\bigO(n^2)$ memory, where $n$ is the number
of vertices of the input polygon.

For the additively-weighted straight skeleton, the same approach can be
used. The computation of potential split event times is slightly more
involved but still in $\bigO(n^2)$.   On speed-change events, a possibly
linear number of collapses have to be recomputed, but the amortized cost
for these is still only linear.  Therefore, the additively-weighted
straight skeleton can also be computed in $\bigO(n^2 \log n)$ time and
$\bigO(n^2)$ space.

Biedl et al.~\cite{Bie&15b} compute the
multiplicatively-weighted straight skeleton (for positive weights) for
monotone polygons in $\bigO(n \log n)$ time and $\bigO(n)$
space.  Their approach hinges on the fact that all mitered offsets of a
monotone polygonal chain are still monotone.  As this property carries
over to additively-weighted input as well, the same algorithm can also
be used to compute additively-weighted straight skeletons for monotone
polygons in time $\bigO(n \log n)$.  This promises to be useful in
practice as footprints of buildings often are monotone polygons.

Aichholzer and Aurenhammer~\cite{AiAu98} describe a triangulation-based
algorithm. Their core idea is to maintain a kinetic
triangulation of the wavefront polygons, and to keep track of triangle
collapses in a priority queue since triangle collapses signal events.  As in
our prior work \cite{PalHH12}, it is this algorithm that we base our
implementation on.  We augmented the priority queue with the times of
speed-change events, and are thus able to compute additively-weighted
straight skeletons.

Our implementation is based on CGAL \cite{CGAL} and is capable of exactly
computing the straight skeleton of a PSLG with non-negative additive and
multiplicative weights. The straight skeletons and offsets in the figures of
this paper were produced by our code.  For the images, our implementation
exported the facets of the roofs and terrains to the graphics software
Blender~\cite{Blender} for rendering.

\section*{Acknowledgements}

Research by Martin Held and Peter Palfrader is supported by Austrian
Science Fund (FWF): P25816-N15.

\section*{References}
\bibliographystyle{elsarticle-num}
\bibliography{hp-add-skeleton}

\end{document}